\documentclass{article}
\usepackage{cs_techrpt_cover,usenix,epsfig,endnotes,url,amsthm}

\title{A model and framework for reliable build systems}
\author{
{\rm Derrick Coetzee}\\
\and
{\rm Anand Bhaskar}\\
University of California, Berkeley \\
\{dcoetzee,bhaskar,necula\}@eecs.berkeley.edu
\and
{\rm George Necula}\\
} 

\begin{document}
\maketitle

\let\thefootnote\relax\footnotetext{Also available as Technical Report No. UCB/EECS-2012-27. All rights to this work are released under the Creative Commons Zero Waiver (CC0). It may be used by anyone for any purpose without permission or condition. This is not a peer-reviewed work. Published 2012 February 17.}

\begin {abstract}
Reliable and fast builds are essential for rapid turnaround during development and testing. Popular existing build systems rely on correct manual specification of build dependencies, which can lead to invalid build outputs and nondeterminism. We outline the challenges of developing reliable build systems and explore the design space for their implementation, with a focus on non-distributed, incremental, parallel build systems. We define a general model for resources accessed by build tasks and show its correspondence to the implementation technique of \emph{minimum information libraries}, APIs that return no information that the application doesn't plan to use. We also summarize preliminary experimental results from several prototype build managers.
\end {abstract}

\section{Introduction}

Large software projects often reach thousands of files and millions of lines of source code. Build automation systems, or \emph{build systems} for short, are responsible for automating the execution of build tools such as compilers in order to process all the source code and produce the final, executable output. The time required to execute a build is a critical factor in a number of software engineering metrics such as: developer cycle time, frequency of continuous integration testing, throughput of check-in verification systems, and time to ship a critical patch; yet a 2003 survey showed that more than half of the 30 surveyed commercial projects had a clean, sequential build time of 5-10 hours.~\cite{ElectricCloudPatent} This motivates the development of builds that can run faster than a clean build.

To address this need, existing build systems provide two features: \emph{parallel builds}, in which multiple build tasks are executed simultaneously, and \emph{incremental builds}, in which results of previous builds are reused and only a subset of build tasks are run, based on what build inputs have changed. In both types of builds, the developer must explicitly specify \emph{dependencies} for each build task, describing other build tasks which must run before it. For example, in a C~project, C~source files must be compiled into object files before the object files can be linked into an executable binary. If even one dependency is omitted, the soundness of both parallel and incremental builds is compromised: build tasks may be run out of order, leading to incorrect re-use of out-of-date results, build failure due to missing results, and race conditions due to concurrent access to files. Whether a failure occurs, and which failure occurs, depends on which input files have changed and the build schedule selected by the build system. As a consequence, ``[m]ost organizations run their builds completely sequentially or with only a small speedup, in order to keep the process as reliable as possible.''~\cite{ElectricCloudPatent} If developers and organizations viewed their parallel, incremental builds as highly reliable, they could use them consistently throughout the development, testing, and release process, accelerating these processes and offloading the mental burden of build management.

Incomplete dependencies arise naturally whenever a developer change introduces a new dependency, but fails to correctly update the dependency information. As a simple example, consider the build described by this makefile:
\begin{verbatim}
all: generated.h foo

generated.h: config
    ./gen config -o generated.h

foo: foo.c
    gcc foo.c -o foo
\end{verbatim}
\noindent
Here, a tool called \emph{gen} is run to generate the header file generated.h from a file \emph{config}; then the binary foo is compiled from the C~source file foo.c. Now suppose the developer modified foo.c to include the header file generated.h, and also modified \emph{config}. A serial build will still produce the expected result, since generated.h is listed before foo in the ``all'' target; but an incremental or parallel build may run the \emph{gcc} action before, or simultaneously with, the \emph{gen} action, leading to incorrect output or build failure.

This work explores background and existing work in build systems and obstacles and design options for reliable build systems. It also presents a formal model for build system analysis and discusses some early experimental results with several prototypes.

\section{Background}

Dependencies in a build are described by a dependency graph, a directed acyclic graph (DAG) where build tasks (typically, invocations of a build tool) are vertices, and an edge from A to B indicates that B depends on A. Given such a graph and a uniform set of processors, deciding which tasks to run at what time is an instance of the \emph{DAG scheduling problem}, which is studied in the context of static scheduling of processes in high-performance computing. It is NP-complete even in the restricted case where there are two processors, no dependencies, and the run time of every task is known (closely related to the \emph{partition problem}), but a number of effective heuristics are available in practice.

A single-node build can be scheduled using a topological sort, which can be computed by a simple online algorithm: at each step, select an arbitrary vertex with no incoming edges to run, and when it completes, delete it. A similar algorithm can schedule parallel builds: whenever at least one processor is free, run an arbitrary task with no incoming edges, and whenever a build step finishes delete its vertex. It is possible that all tasks have incoming edges, in which case processors may remain idle until more tasks complete. This algorithm, used by \emph{make}, is a version of Graham's classical online list scheduling algorithm,~\cite{graham} and has the advantage of not requiring task runtimes, but does not take into account the critical path (the path of largest total time).

The technique can be improved by assigning priorities to nodes, using any of a number of heuristics, and then selecting the node with the highest priority at each step.~\cite{Wu01efficientlocal} Effective priority assignment requires task runtime estimates, which can be inferred from previous builds and/or a runtime model. This approach has not been yet tried.

\subsection{Shared state and resources}

To model builds we define the \emph{shared state space} $S$, typically representing the filesystem and other state visible to multiple tasks as well as the task input (e.g. command line, environment). A \emph{resource} is a function $r$ with domain $S$. Intuitively, a resource is anything that may be returned by a library function. Resources can range from simple predicates (``does this file exist?'') to values (``what are the contents of the file at this path?'') to complex operations (``what is the abstract syntax tree obtained after parsing the source file at this path?''). A resource can also encompass many files (such as the contents of all files in a subdirectory). Prior to starting the build, a fixed (typically infinite) \emph{resource space} is selected---no build process may access resources outside that set.

A build task performs a sequence of \emph{accesses} (reads or writes) to resources. During a parallel build, accesses by many tasks may be interleaved to form an access sequence, subject to the constraint that if $g$ depends on $f$, all accesses by $f$ precede all accesses by $g$. Reads make the current value of a resource accessible to the task executing it, while writes update the shared state in such a way that one or more resources are set to a new given value. Any resources not written to during a write must remain unmodified.

Build tasks must be \emph{deterministic}, in the sense that their accesses (including type, resource, and value written) depend only on the results of prior reads. Two tasks are said to \emph{conflict} (during a particular build) if one of them writes a resource that the other reads or writes. A given build is \emph{valid} if, for any pair of conflicting tasks, there is a directed path from one to the other in the dependency graph. It can be proven that if a given build is valid, it produces the same final result as any other parallel schedule, given the same initial shared state (see appendix~\ref{validbuildproof}). This allows us to meaningfully define a \emph{valid configuration} as a pair (dependency graph, start state) that produces valid builds.

To model an incremental build, suppose we start with initial state $s_i$, perform a build resulting in state $s_f$, modify the shared state to get $s'_f$, and then perform another build. For now, we assume that for every task $f$, $f$ has no effect when acting on $s_f$ --- that is, right after the first build is complete, re-running any one step will change nothing (in practice, this typically means retaining and not reusing intermediate files). Define the special task $d$ updating state $s_f$ to $s'_f$, representing the actions of the developer, and add edges from $d$ to all tasks that $d$ conflicts with. Now we assign $d$ the lowest priority and create a DAG schedule. This will move all nodes that don't conflict with $d$ before $d$, where they will have no effect, since they are acting on $s_f$. Effectively, this means the only part of the graph that needs to be scheduled is the transitive closure of $d$.

\subsection{Selecting a resource space}

There is a tradeoff in the choice of the resource space: if resources encompass too much state, there will be spurious conflicts. For example, a trivial resource space has a single resource returning the entire shared state. In this space, all reads conflicts with all writes, and the build must run sequentially.

On the other hand if resources are too fine-grained, the result will be that processes read and write a very large number of resources, resulting in excessive overhead for build management and a large dependency graph. For example, if every byte of every file had its own resource, a typical build task would access many thousands of resources.

One straightforward strategy is to create a single resource for the contents of each file on the disk. To account for the creation and deletion of files, there is a resource for every possible filepath, with a special value indicating the file does not exist or is inaccessible, analogous to a ``read file contents'' library function that returns NULL on failure. This simple resource system is similar to that used by \emph{make} and is sufficient for many builds.

Many applications require a notion of a collection/set resources, such as a directory. A naive representation would have a resource for the contents of each collection; but then two tasks creating files in the same directory would conflict. Such a collection is best represented as an infinite set of resources, one for each potential element of the collection, indicating whether or not that element is present (in the case of a directory, one for each filename, indicating whether that file exists in that directory). A process that reads the collection (e.g. listing the files in the directory) reads all of these resources (note that this requires a concise representation for certain infinite resource sets). A process that adds or removes items from the collection may only affect a few of them.

Although files are by far the most common resource, there are many examples of other resources that are useful. For example, the Linux kernel build has a single header containing all configuration options which is included by all source files. In order to make incremental builds useful in the event of configuration option changes, the Linux build tracks each option as a separate resource.

A set of resources in a resource space may be \emph{contracted} to form a merged resource which yields a tuple of all the resources used to form it. Such contracted resources allow a gradual tradeoff between the number of resources accessed and the number of conflicts that occur during the build---see section~\ref{taskandresourcegranularity} for more details.

\subsection{Hidden resources}
\label{hiddenresources}

There are resources that are used in practice by many tools but are not tracked by existing build managers, either by convention or because supporting them is difficult. These include:

\begin{itemize}
\item Compiler flags and tool configuration: if a build is done, and then tool configuration is altered, for example to enable debugging flags, all files must be rebuilt. If it is changed back, there is no need to rebuild everything again. Visual Studio implements solution configurations with separate output directories to cope with this, but these are rarely used for more than two configurations. Vesta~\cite{VestaBook} records outputs of many previous builds in its derived file cache.
\item Nonexistent files: A C source file reading "\texttt{\#include <stdio.h>}" will search the system include path in order to find the header. Developers often add project directories to this path. If a file named "stdio.h" were ever created along this path, it would change the result of the task, but most extant tools would not detect the need to rebuild. Vesta~\cite{VestaBook} and \emph{scons}~\cite{scons} track dependencies on nonexistent files.
\item Build configuration file: determining which part of a build needs to be rebuilt after changing the build configuration file itself (e.g. Makefile) is a difficult problem. Even small changes may affect all tasks or only a few, and determining which may require analyzing structural changes since the previous version.
\item Build tools, libraries, and system headers: upgrading build tools or libraries used by build tools, or copying a source tree to a machine with different tools, can dramatically alter build output, but these are usually untracked. Sometimes this results in an incompatible combination of files generated by different versions of tools. This motivates the common industry practice of including all build tools in the version control repository. As mentioned in section~\ref{taskandresourcegranularity}, it often makes sense to treat these files as a single aggregate resource.
\item Non-file resources: accesses to network resources, peripheral devices, the time, and so on are usually untracked. Some real-world builds retrieve files during the build from remote sources, query remote databases, or even do web service queries. These should be tracked as resources, even if coarsely.
\item Special files: some files like those under "/proc" and "/dev" in UNIX may fail to update their last modified time, or even change each time they are read.
\item Operating system: the results of system calls made to the kernel by build tools may affect build output. These results may vary depending on the specific operating system, operating system version and patches, filesystem and drivers, or even kernel configuration options. These are untracked by all extant systems, and largely benign given a carefully designed resource space and a standards-compliant operating system.
\end{itemize}

The choice of how to handle hidden resources depends on the resource space, the application, and build platform variability. Some applications may not use certain types of resources or may be built only on a fixed build server. In some cases, like the build configuration file, merely detecting any change and triggering a full rebuild may be sufficient in practice. In other cases, where changes are frequent, fine-grained resource tracking is needed.

\section{Related work}

\subsection{Build systems}

A small number of build systems dominate in practice today, most of them based on \emph{make}, created by Stuart Feldman in 1977 at Bell Labs.~\cite{make} With \emph{make}, the developer uses a domain-specific language to specify a series of targets, and each target may declare explicit dependencies on other targets and/or source files. Each target has an associated shell command that builds the target. This explicit representation of the dependency graph facilitates both incremental and parallel builds. However, dependencies must be specified correctly; if they are not, incremental builds may fail to rebuild portions of the application, leading to incorrect results with unpredictable behavior, and parallel builds may produce different outputs nondeterministically. Make is designed for use on a single machine, and build results are not shared between developers. A number of important dependencies are either difficult to represent or omitted by convention, such as the ones mentioned in section~\ref{hiddenresources}---changes in these may require a complete rebuild. Even incremental builds in \emph{make} take time proportional to the size of the build as a whole due to the need to process all targets and scan all input files for changes. This process can be accelerated by using file timestamps to detect changes, at the expense of correctness, since this is not reliable in general. Although some build systems like Apache Ant and MSBuild adopt XML build description files in place of \emph{make}'s domain-specific language, facilitating greater extensibility, they still inherit all of these issues.

One of the most developed research build systems is Compaq/Digital Systems Research Center's Vesta, developed in the late 1990s and released under the GNU LGPL in 2001.~\cite{VestaBook} Although Vesta does not support parallel builds, it provides incremental builds reliable enough to be used in practice for product releases (``every build is incremental''). It tracks dependencies that extant tools like \emph{make} incorrectly ignore, such as dependencies on build description files, compiler flags, nonexistent files, and build tools. Through the derived file cache, compilation outputs are easily reused between developers. Change detection and inferrence of dependencies is implemented using a custom filesystem, so that the filesystem does not need to be scanned to find modified source files, and a sophisticated functional build description language allows large portions of the build to be reused.~\cite{VestaCaching} Using its derived file cache, Vesta can reuse results not only from the previous build but from all previous builds, by treating tool executions as functions and memoizing their results (see their \emph{runtool cache}).

Vesta was deployed by large product teams at Compaq and Intel, but has not achieved widespread use. This can be attributed to several factors. One is that Vesta is a ``package deal,'' requiring teams who use it to also use Vesta's custom filesystem and version control, both of which are not as mature, featureful, or well-supported as existing systems. Migration of existing projects to Vesta while preserving change histories is difficult or impossible, and requires translating existing build description files into Vesta's very different language. Modern builds are done in parallel, even on single nodes, and large builds are done on clusters, neither of which Vesta supports. Finally, the cost of incorrect incremental builds is hidden: it is difficult to measure the time spent by developers resolving incorrect builds, or the time that might have been saved by building product releases incrementally.

A central feature of Vesta was \emph{repeatability}, in which all source files used in a particular build can always be retrieved at a later time, and used to repeat the same build. Although this feature is valuable (e.g. for isolating source changes leading to behavior changes), it is separable from the other features and depends critically on integration with version control, so it is disregarded in this report.

A very different approach to build systems was taken by Electric Cloud,~\cite{ElectricCloudPatent} which disregarded incremental builds in favor of using clusters of machines with parallel processors to speed up full builds as much as possible, currently deployed as an enterprise commercial product. A network filesystem infers dependencies, and visualization tools facilitate the identification of bottlenecks. Although fast and well-supported, Electric Cloud is not suitable for routine developer builds, does not scale down effectively to small projects, and is too expensive for many applications such as open-source development.

More recently, in 2012, Electric Cloud has released ElectricAccelerator Developer Edition,~\cite{ElectricCloudDE} which is designed to run on a single machine, infers dependencies, and implements accurate incremental and parallel builds, scaling up to four cores. Although this product effectively accomplishes the primary goals set out in this report, it chooses a single design and leaves room for improvement in numerous directions, such as tool cooperation, sharing of derived files, custom resources, and so on.

\subsection{Build augmentation}

A number of more practical efforts have sought to augment existing build tools by providing services to accelerate them or improve their reliability.

The GNU Make manual illustrates how to use the ``\texttt{-M}'' flag of gcc (the GNU C Compiler) to generate \emph{make} dependencies for C/C++ builds on-the-fly and keep them up-to-date automatically. These dependencies are incomplete, including only header and source files, but greatly increase reliability and reduce maintenance effort compared to manual specification for this specific type of build.

The \emph{ccache} tool,~\cite{ccache} based on \emph{compilercache},~\cite{compilercache} caches results of invocations of standard compiler tools like \emph{gcc}, even if the intermediate files are later deleted or overwritten. It can dramatically improve incremental build times for C/C++ projects, but does not generalize to other tools. \emph{scons}~\cite{scons} provides similar functionality.

Google relies on conventional distributed builds with coarse-grained tasks and manually-specified dependencies. Their efforts have focused on dramatically reducing the runtime of important build tools, such as the C/C++ linker, which is a bottleneck in large parallel builds because it is used in the final step to combine all results.~\cite{goldlinker}

\section{Build specification}

Systems like \emph{make} lean heavily on build specification via an explicit dependency graph. This has certain advantages: dynamic scheduling of incremental, parallel builds is straightforward as outlined above, and it's also intuitive to create build description files that include multiple targets and allow the developer to choose to build only a subset of them (and these targets may share dependencies).

One of the simplest ways to specify a build is with a sequence of shell commands, a basic shell script. Any sequential build is equivalent to such a script. Both incremental and parallel builds can be implemented in this setting by inferring dependencies from previous builds (see sections~\ref{offlinedependency}, \ref{transactions}). This scheme can be extended to include nonrecursive function calls and variables without adding significant complexity. It has the advantage of being intuitive and familiar to procedural programmers, but unlike explicit dependency graphs becomes less intuitive when building a subset of targets.

The most general type of build specification is the build program or build script. Such a script is written in a general-purpose language and may employ sophisticated abstraction mechanisms, algorithms, and data structures. Vesta's functional build language~\cite{VestaCaching} and \emph{scons}'s Python build descriptions~\cite{scons} are examples. In some cases it may even be integrated into the application being built, allowing the application to generate source code and rebuild itself or portions of itself. Incrementalism can be extracted using memoization, as in Vesta, and parallelism can be extracted using futures. Although the most flexible option, automatically extracting incremental and parallelism from a general build program is challenging and in some cases infeasible.

Some practical tools mix these approaches; \emph{make} for example incorporates basic variables and conditionals while remaining primarily based on dependency graphs. Other hybrids may be possible, such as a Makefile-like language where both dependency lists and actions can be program fragments in a general-purpose language. A major goal of future work is to design a build description language that can concisely represent typical builds in practice, minimizing opportunities for error, but remain flexible and scalable enough to accommodate large and complex builds.

\section{Capturing access to shared state}
\label{capturing}

Standard tools such as \emph{make} rely on the developer to manually specify all shared state which is accessed by each task, making the system unable to distinguish a valid build from an invalid one. There are several techniques for reliably, automatically capturing access to shared state.

\subsection{File system filtering}

It is straightforward to implement a filesystem or network filesystem server which acts as a proxy, monitoring all file operations and mapping them onto an underlying filesystem. Some filesystem subsystems, as in Windows NT, have explicit support for filters to capture all file operations, for use by virus scanners and backup utilities. To detect conflicts, the system must know which build task is performing each file operation, usually inferred from the process ID. The technique extends easily to distributed build systems.

This approach was the primary means of capturing dependencies in Vesta, and is simple to deploy (although it typically requires superuser access). Its main disadvantages are that it only captures operations on files and only at whole-file granularity, it must be applied to every filesystem a build process could possibly access, and that the file API is typically at an inappropriate level of abstraction, yielding too much information on each call.

\subsection{System call interception}

On typical MMU-based systems, all access to shared state by a process passes through system calls, which can be intercepted either through binary rewriting or through kernel support for system call interception such as \emph{ptrace}. Unlike file system filtering, system call interception can capture all access to shared state including all filesystems, the network, and kernel data structures (with some minor exceptions like RDTSC, which can be disabled). 

One obstacle with system call interception is that typical build tools generate very high volumes of system calls, many of which are unimportant for dependency tracking. In experiments, handling all system calls with a central \emph{ptrace} monitor process led to crippling overhead. Binary rewriting suffers from a different performance issue: load-time rewriting is too expensive for short-lived processes, necessitating on-disk caching of instrumented binaries. Kernel patches (for ptrace) or in-process filtering (for binary rewriting) can reduce the number of system calls, but is more difficult to implement and deploy and less flexible than minimum information libraries.

A more fundamental obstacle with system call interception is that applications routinely invoke system calls that return more information than they require. For example, UNIX applications testing for the existence of a file routinely use the \emph{stat()} system call, which also returns the last modified and last accessed time of the file, which change frequently. Another daunting case is environment variables, which are passed to new processes as a complete array; there is no way to determine which ones are used through the system call interface. Similarly, an application may read in a database file just to use one row of a table, or (as was observed in some open-source tools) cache the contents of a directory to accelerate future queries. To ensure correctness, the build system must assume all the information available to the process is used by it, which leads to unacceptable performance. Dynamic taint tracking,~\cite{dynamictainttracking} used to track the flow of untrusted data in security applications, could be used to trace the flow of system call results in-process, but has high overhead, and may fail to accurately track complex cases, such as an array of environment variables being transformed into a hash table data structure.

\subsection{Minimum information libraries}

A \emph{minimum information library} is a library designed to supply the minimum information that will be used by the caller and no other information, even in case of error. For example, whereas a POSIX application may use \emph{stat} or \emph{fopen} to determine if a file exists, a minimum information library would supply a \emph{fileExists} method returning a boolean. It would only return true if the file exists, or false if it doesn't exist or is inaccessible. Similarly, environment variables would be accessed through \emph{get} and \emph{set} functions instead of by parsing the environment block. These expose fine-grained dependencies in the application while still making the same number of system calls under the covers.

Minimum information libraries have a natural correspondence to resources as defined in this work: every resource can have an associated call in the library that reads and (where applicable) writes that resource. Other calls may read or write multiple resources.

A minimum information library can be easily instrumented to acquire one or more resources with every call, or to acquire a single resource to serve many calls, avoiding a proliferation of acquisitions. It can either save this information for later analysis, or contact a central build manager process to acquire a lock on the resource. By eliminating or wrapping all library calls that invoke the kernel, all access to shared state can be directed through the minimum information library, ensuring that all dependencies are systematically tracked.

When an application is written against a minimum information library, dependency tracking is simplified, but for many build tools that are either binary-only or managed by third parties, porting to another runtime library is a poor investment. For cases like these, a promising alternative is the \emph{build wrapper}, a small tool using a minimum information library that replaces the tool and acquires any needed resources, then invokes the underlying tool normally. Such a wrapper often requires only a small subset of the functionality implemented by the full build tool.

Unlike the other solutions above, minimum information libraries require some work to be done for every build tool, including application-specific build tools, and bugs in this code can lead to build unreliability. However, the number of build tools in a build is very small compared to the number of build tasks, typically ranging from 1 to 50. For widely-used tools like \emph{gcc}, the work can be shared among many users of the tool and developed to maturity, while application-specific build tools tend to be very simple, with dependencies inferrable from the command line alone.

\section{Change detection}

The change detection problem is the problem of capturing changes to shared state \emph{between builds}, for the purpose of implementing incremental builds. Traditional build systems like \emph{make} rely on comparing timestamps between task input and output files to determine if a task needs to be re-run. This is overconservative, in that unmodified files may have updated timestamps; incorrect, in that tasks may not be run if timestamps travel backwards (as when restoring from a backup); and inefficient, in that all tasks and all their input files must be examined even for a small incremental build. New build managers like \emph{scons}~\cite{scons} rely on hashes of file contents to detect changes, fixing the first two problems at the expense of even more inefficiency. Moreover, both these approaches are ineffective for resources other than simple files.

Ideally, change detection should log exactly which resources in the chosen resource space are modified at the moment they are modified, making their retrieval trivial. This would be straightforward if all applications were written against the same minimum information library as the build tools, but this is infeasible in practice because development tools are generally third party and difficult to wrap due to being interactive and long-lived.

Some kernels support keeping a log of all modified files, including NTFS's USN change journal~\cite{USN} and Linux with Stefan Büttcher's fschange patch.~\cite{fschange} Combined with a resource database that tracks old values of resources, these can be used to detect changes to filesystem-based resources as soon as they occur. ZFS uses Merkle trees to efficiently track hashes of the contents of all files at all times, for integrity and de-duplication, but this information is not user-accessible without a patch. Network-based resources can be intercepted by packet sniffers, at some overhead.

\section{Specifying and inferring dependencies}

\subsection{Manual dependency specification}

Although primitive, manual dependency specification offers a transparency and flexibility difficult to achieve with other methods. If coarse-grained tasks are used (see section~\ref{taskandresourcegranularity}), dependencies don't have to be updated too often, easing the maintenance burden. In this scenario, the primary function of the build manager is to detect invalid builds, with error messages suggesting how to repair the build description file. It can also optionally warn about redundant dependencies.

\subsection{Phased dependency specification}
\label{phaseddependency}

An extension of manual dependency specification is to have a build that proceeds in phases, where earlier phases generate dependencies used by later phases. A simple example of this is the typical integration of \emph{make} with \texttt{gcc~-M}, where dependency files are generated from source files in the first phase, and in the second phase source files are compiled using those dependencies. This can be extended to more phases in scenarios where tools must first be built to generate dependencies. Because each phase can be parallelized and incrementalized separately, this approach can be similar in performance to the manual approach. Some degree of interleaving may be possible, but caution is required to ensure that no dependencies become available after the point where they are needed (or alternatively, rollback may be used in this case---see section~\ref{transactions} below).

\subsection{Offline dependency graph augmentation}
\label{offlinedependency}

An alternate strategy is to infer dependencies based on the conflicts observed in an invalid build. If two tasks conflict but there is no directed path between them, the system can add an edge between them, but needs more information to infer the direction of the edge. One simple way to supply this information is to give a serial ordering of all tasks---then if A and B conflict, whichever comes earlier in the serial order is run first. In the case of dynamically scheduled tools, such a serial order can be inferred after the fact from any deterministic walk over the task execution tree of the build. Once the graph is updated, the build is re-executed (invalidating the conflicting tasks to force them to re-execute), and this process is repeated until a valid build is observed. Termination is guaranteed because eventually the dependency graph will contain a path through all tasks, and so necessarily be valid.

Inferred dependencies are stored as derived files that can be shared between developers (via a derived file cache, or simply through version control). For this reason, invalid builds are expected to occur infrequently, only when source files change in a way that adds dependencies.

Because the serial ordering is used to direct dependencies, the parallel build that results from this algorithm will produce the same final result as a sequential build of the serial ordering (per the theorem of Appendix~\ref{validbuildproof}). Such a build is predictable, easy to test, and easy to conceptualize for the developer. Compared to manual dependency specification, dependency inferrence allows more concise build description files that require less frequent updating. However, unforeseen conflicts may lead to excessive edges and build bottlenecks.

A challenging problem for this strategy is determining when to remove inferred dependencies. The build can easily detect when there is no conflict between two tasks, but it is difficult to establish whether the lack of conflict is a short-lived or long-lived phenomenon. For example, in a C++ project, there may be a certain header file which is only included in debug builds, resulting in dependencies that appear in debug builds but not in release builds. One simple strategy is to periodically erase all inferred dependencies and re-run the build to reproduce them.

\subsection{Transaction-based task synchronization}
\label{transactions}

Another strategy is to prevent any invalid builds from occurring by inferring dependencies on-the-fly at runtime. Using concepts from database transactions, we lock resources before accessing them by submitting a lock request to the build manager process. If the resource is already locked, the task is blocked until it is available. Tools with build wrappers can lock all necessary resources before invoking the real tool. However, once locks are in use deadlock is possible, and to make progress tasks must support abort and rollback, which kills the task and undoes its previous effects to the shared state.

By itself, this algorithm will yield an unpredictable ordering of conflicting tasks, leading to nondeterminism in build outputs. Suppose we wish instead to produce the same final output as the sequential serial build. In this case, we can employ a version of multiversion timestamp concurrency control~\cite{MultiversionTimestampCC}, placing each task inside a transaction with a virtual timestamp equal to its order in the serial build. If a task observes a value that was written by a task with a later timestamp, this is termed \emph{physically unrealizable behavior}, and forces an abort and rollback of the reader and any tasks influenced by its writes directly or indirectly (ordinary multiversion timestamping rolls back the writer, but in our scenario this can lead to a failure to make progress). Unlike the pessimistic locking strategy above this is an optimistic strategy, and so avoids blocking tasks at the cost of more frequent restarts.

\section{Task and resource granularity}
\label{taskandresourcegranularity}

Fine-grained tasks allow incremental builds to avoid redundant work and parallel builds to run more tasks in parallel. Generally the most fine-grained task possible is an execution of a build process, since such tasks cannot be easily subdivided. However, the intuitive association of a single process with a task may be counterproductive: a large number of processes leads to a large number of tasks and a large dependency graph which takes more time to construct and analyze. By partitioning this graph and collapsing each partition to a single task, the graph size can be dramatically reduced with only a modest increase in incremental build times. There is also little to no decrease in parallelism in practice, either because the reduced build is still capable of saturating the hardware's parallelism capacity, or because individual build tools support parallel execution. One typical strategy for accomplishing this is switching from a ``file-based'' compilation method to a ``module-based'' method, where entire directories are compiled into static/shared libraries or binaries in a single step. Some build tools, like the Microsoft Visual C\# compiler, exclusively use this approach.

Along with a decrease in graph size, the frequency of updates to the dependency graph is lowered, making manual graph maintenance more feasible and leading to a smaller number of rebuilds.

Similarly, the intuitive fine-grained association of resources with individual files can be counterproductive. For example, every task has a set of ``owned'' resources that only that task depends on, which can be collapsed into a single resource without increasing build times. If the tasks are coarse-grained, this can substantially reduce graph size. Another important case is the set of system resources, such as build tool executables, that are rarely updated and used by nearly all tasks. By collapsing rarely-updated, widely-used resources into a single resource, an enormous number of dependency edges are eliminated, and long incremental builds are only needed during a system update---at which time a full rebuild is needed anyway.

Decreasing graph size decreases overhead differently depending on the system used. In a lock-based system with a central build manager, it results in less lock and unlock operations and less interprocess communication. In a system that logs dependencies, it leads to fewer and smaller log files and less time loading them. In a system that performs static DAG scheduling, the scheduling algorithm runtime is reduced and an improved schedule may become feasible. These optimizations are essential to ensure that build overhead does not dominate build time.

In order to achieve these gains, the partitioning must be known and available to all tasks before the build begins. Both task and resource partitioning can be inferred by analysis of the dependency graphs of previous builds. Task partitioning can also be specified implicitly by describing each task using a command sequence or script that performs all necessary actions for that task. Resource partitioning can be specified manually, e.g. by using directory patterns to distinguish application and system resources. A promising hybrid approach that both limits incremental build time and keeps graph size small is to automatically use smaller partitions for resources that are modified frequently (e.g. the module the developer is currently working on) and increasingly larger partitions for resources that are modified less frequently.

\section{Preliminary experimental results}

Three prototype build systems were constructed.

In the first, a ptrace-based prototype that could only perform full builds, a pessimistic locking scheme was used where build processes took locks on any files they accessed. Processes also took ``predicted locks'' on any files they accessed during previous builds; predicted locks cause processes with later timestamps (which occur later in the serial build) to block if they attempt to lock the file. This allows cascading rollback to be avoided. Given enough concurrent processes, this build scaled to 85\% the time of a parallel \emph{make} build of the Linux kernel. However, it was not a complete system, as it was unable to handle unexpected new dependencies, could not perform incremental builds, and inferred its list of processes to execute from a prior \emph{make} run, making it necessary to rerun the \emph{make} build whenever this process sequence was changed.

The major performance bottleneck in this prototype was the necessity for the central build monitor process to sequentially handle all ptrace messages. A variant of this prototype used binary rewriting based on Jockey~\cite{jockey} to track system calls without the use of ptrace. Jockey rewrites binaries at load time by searching for system calls, and also keeps a cache of patches to apply for binaries it's seen before. In practice, even with caching, the system added too much overhead to be practical due to the Linux kernel build's enormous number of short-lived processes like \emph{cp} and \emph{mkdir}. This is less likely to be an issue in a more monolithic build system.

The second prototype was based on multiversion timestamping and was able to handle process hierarchies. Instead of replacing \emph{make}, \emph{make} is run sequentially and children of \emph{make} are run speculatively, pretending to succeed so that \emph{make} will continue and begin the next process. Rollback was implemented by performing all writes in a temporary location and then committing them after a process completes, which can be accomplished by rewriting results of system calls (a simple form of filesystem virtualization). Although the system was able to run real-world builds, and was powerful enough to complete builds even given no initial dependency information at all, the overhead of its transaction management and filesystem virtualization prevented it from scaling to larger builds, particularly since the build manager ran sequentially. On very small builds with few dependencies, it could outperform a sequential \emph{make} build by 30\% while offering the same results and reliability, but even on medium-sized builds this performance advantage was lost. In neither case could it compete with parallel \emph{make} builds.

The system also supports reliable incremental builds: it keeps a cache similar to Vesta's \emph{runtool cache}, and whenever a process is re-executed with the same inputs as in a prior run, it skips running the process and commits its cached results. Although its incremental builds are much faster than its full builds, they are not competitive with incremental builds by \emph{make}, for several reasons: the main \emph{make} process is still run as it would be in a full build, input files have to be hashed to implement the cache reliably, and the filesystem virtualization (particularly committing cached results) is expensive.

The third prototype abandoned transactions and system call tracing in favor of cooperation with build tools. A variety of open-source build tools were instrumented to declare their dependencies at runtime using a C library called \emph{deptracker}, which then wrote them out to an XML file when the process exited. An offline analysis step would then load all of these, detect conflicts, and (together with sequential build order information logged by an instrumented \emph{make} tool) generate a supplementary Makefile to augment the existing dependency graph. Initial performance evaluation with small builds showed that the time needed to load and process the XML files was a substantial portion of build time, as much as 30\% of the build, suggesting that a coarser granularity of tasks and/or resources is needed to accelerate this stage.

Another challenge for this prototype was the impracticality of maintaining a forked and instrumented codebase for every build tool used by a build, including many like \emph{gcc} with much larger builds themselves than the build under evaluation---effective build wrappers could mitigate this problem.

\section{Conclusion and future work}

This work discussed design options for constructing a reliable build system and highlighted tradeoffs between them, but many of the ideas remain untested. A clear next step is building a complete build manager that can handle a real-world large build, including change detection and dependency inferrence, and measure overhead compared to existing solutions. Developing a meaningful performance testing method for incremental builds is another challenge. Expanding the model and giving design options to support distributed builds would be valuable.

Incorporating features of Vesta, such as a shared derived file cache and repeatable builds via integration with existing version control systems would be another intriguing direction. Taking this to extremes, it may be valuable to have a ``cloud cache'' that shares derived files for building open-source projects among developers throughout the world.

Some of the concepts that are useful for reliable build systems can also be applied in other domains. For example, because minimum information libraries allow resource dependencies of code segments to be reliably and precisely identified, they can be used to compute information transfer from one portion of a program to another through shared state, which is often overlooked by dynamic analysis tools.

Finally, there is a great deal of practical work needed to get a functional reliable build system into the hands of everyday users, including supporting major tools and environments, providing an expressive build description language, and pushing for better change detection support in mainstream kernels.

\section{Acknowledgements}

The authors wish to thank: the Compaq research team, for providing a fundamentally new design in the space of build systems; David Wagner, for advice regarding coping with concurrent accesses to files and for suggesting other faculty; Maria Welborn, for advice regarding system call virtualization via system call rewriting; Ras Bodik for providing assistance with funding; Philip Reames, for ideas on reusing concepts in other domains; Eric Brewer, for providing suggestions regarding file operation interception and feedback on evaluation; and software developers at Microsoft, UC Berkeley, and other organizations for providing feedback regarding their experiences with build systems.

{\footnotesize \bibliographystyle{acm}
\bibliography{apmake_techreport}}

\appendix
\section{Well-definedness of a valid configuration}
\label{validbuildproof}

A \emph{configuration} specifies the dependency graph and initial shared state for a build. Recall that a valid build is one where, for any pair of conflicting tasks, there is a directed path from one to the other in the dependency graph. We begin by showing a lemma:

\newtheorem{theorem}{Theorem}[section]
\newtheorem{lemma}{Lemma}[section]

\begin{lemma}
If a given build is valid, any other valid build with the same configuration produces the same final result.
\label{validbuildsameresult}
\end{lemma}

\begin{proof}
Define the canonical access sequence as the sequence obtained by fixing some topological order and executing each task sequentially in that order. Given a valid build's access sequence, we will perform a series of swaps to transform it into the canonical sequence.

Suppose two tasks are interleaved (neither performs all its accesses before those of the other). Then there is not a directed path between them in the dependency graph, and since the build is valid, they must not conflict. Hence we can safely swap accesses to ensure that the two tasks are no longer interleaved. By doing this for all pairs of tasks, we get a sequential schedule which performs all of each task in some order $(t_1,t_2,\ldots,t_n)$ which is a topological sort of the dependency graph.

Any two tasks in a topological sort can be swapped unless there is an edge between them, and the result is still a topological sort. Such swaps can be used to transform the sequence into any other topological sort while preserving the final output, including the canonical sequence. Hence any valid build's access sequence produces the same final result: the result produced by the canonical sequence.
\end{proof}

We now generalize this to the stronger result:

\begin{theorem}
If a given build is valid, all builds with the same configuration are valid and produce the same final result. If a given build is invalid, all builds with the same configuration are invalid.
\end{theorem}

\begin{proof} By Lemma~\ref{validbuildsameresult}, if all builds are valid, they all produce the same final result. It remains to show a single configuration cannot generate both a valid and an invalid build.

Suppose we have an invalid build $(a_1,a_2,\ldots,a_n)$ and a valid build $(b_1,b_2,\ldots,b_n)$, both with a given access sequence. We will gradually transform the first into the second.

We find the first point at which they diverge $a_i \ne b_i$, locate $a_j$ such that $a_j=b_i$, and move it up to the $i$th position by a series of swaps. If $a_j$ did not conflict with any of $a_{i+1},\ldots,a_{j-1}$, then the behavior of all tasks is preserved: the new access sequence is a feasible build, and is valid if and only if the previous sequence was valid.

Suppose on the other hand $a_j$ does conflict with at least one of $a_{i+1},\ldots,a_{j-1}$; let the first be $a_m$. Because swapping $a_j, a_m$ may change task behavior, the build must be conceptually re-executed starting after $a_m$ to get a feasible new access sequence. In the previous iteration, $a_j$ followed $a_m$, whereas in the current iteration $a_j$ precedes $a_m$; this implies the two tasks owning these accesses have no directed path between them in the dependency graph. But $a_j, a_m$ conflict, so the new build is invalid.

In either case, the common prefix of the two builds grows by at least one access with each iteration, and eventually the build $(b_k)$ is reached. However, in both cases the invalidity of the original build is preserved, so $(b_k)$ is invalid as well. This is a contradiction, so there cannot be both an invalid and a valid build.
\end{proof}

This means the definition of a \emph{valid configuration} as a pair (dependency graph, start state) producing valid builds is well-defined.

\end{document}